\documentclass[11pt,a4paper]{article}
\pdfoutput=1
\usepackage{amssymb,amsmath,amsthm}
\usepackage[top=2.2cm,right=2.9cm,left=2.9cm,bottom=2.2cm]{geometry}
\usepackage[utf8]{inputenc}
\usepackage[english]{babel}
\usepackage{ogonek,slashed}

\allowdisplaybreaks[4]

\newtheorem{prop}{Proposition}
\newtheorem{thm}[prop]{Theorem}
\newtheorem{lemma}[prop]{Lemma}
\newtheorem{cor}[prop]{Corollary}
\newtheorem{df}[prop]{Definition}
\newtheorem{rem}[prop]{Remark}
\newtheorem{ex}[prop]{Example}

\newcommand{\Q}{\mathbb{H}}
\newcommand{\R}{\mathbb{R}}
\newcommand{\C}{\mathbb{C}}
\newcommand{\inner}[1]{\left<\smash[t]{#1}\right>}
\newcommand{\tr}{\mathrm{Tr}}
\newcommand{\zero}{\phantom{\,0\,}}
\newcommand{\Z}{\mathbb{Z}}
\newcommand{\U}{\mathrm{U}}
\newcommand{\SU}{\mathrm{SU}}
\newcommand{\id}{\mathsf{id}}
\newcommand{\sbinom}[2]{\genfrac{[}{]}{0pt}{}{\,#1\,}{\,#2\,}}
\newcommand{\D}{\slashed{D}}
\newcommand{\Cl}{\mathcal{C}\ell}

\begin{document}

\title{\vspace*{-1cm}The Standard Model in Noncommutative Geometry\\ and Morita equivalence}

\author{\rule{0pt}{15pt}Francesco D'Andrea$^{\,1,2}$ and Ludwik D{\k a}browski$^{\,3}$\\[15pt]
\textit{\small $^1$ Universit\`a degli Studi di Napoli Federico II, P.le Tecchio 80, I-80125 Napoli}\\[2pt]
\textit{\small $^2$ I.N.F.N. Sezione di Napoli, Complesso MSA, Via Cintia, I-80126 Napoli}\\[2pt]
\hspace*{-5pt}\textit{\small $^3$ Scuola Internazionale Superiore di Studi Avanzati (SISSA), via Bonomea 265, I-34136 Trieste}}

\date{}

\maketitle

{\renewcommand{\thefootnote}{}
\footnotetext{%
\hspace*{-6pt}\textit{MSC-class 2010:} Primary: 58B34; Secondary: 46L87, 81T13. \\
\hspace*{10pt}\textit{Keywords:} noncommutative geometry; spectral action; Standard Model.}}

\begin{abstract}\noindent
We discuss some properties of the spectral triple $(A_F,H_F,D_F,J_F,\gamma_F)$ describing the internal
space in the noncommutative geometry approach to the Standard Model, with $A_F=\C\oplus\Q\oplus M_3(\C)$. We show that,
if we want $H_F$ to be a Morita equivalence bimodule between $A_F$ and the associated Clifford algebra, two
terms must be added to the Dirac operator; we then study its relation with the orientability
condition for a spectral triple.
We also illustrate what changes if one considers a spectral triple with a degenerate
representation, based on the complex algebra $B_F=\C\oplus M_2(\C)\oplus M_3(\C)$.
\end{abstract}

\bigskip


\section{Introduction}
In the spectral action approach to (quantum) field theory, the space of the theory
is the product of an ordinary spin manifold $M$ with a finite noncommutative space 
(cf.~\cite{CM08,vS15} and references therein).
States of the system are represented by unit vectors in $L^2(M,S)\otimes H$, where $L^2(M,S)$
are square integrable sections of the spinor bundle $S\to M$ and $H$ is a
finite-dimensional Hilbert space representing the internal degrees of freedom
of a particle.
The algebra containing the observables is 
the tensor product of smooth functions $C^\infty(M)$ on $M$
with certain finite dimensional algebra $A$. More precisely, one has an
``almost commutative'' geometry described by a product
of spectral triples, with Dirac operator constructed from the Dirac operator of $M$ and certain selfadjoint operator (a Hermitian matrix) $D$ on $H$.

A deep algebraic characterization of the space of Dirac spinor fields $L^2(M,S)$
on a spin manifold
is as the Morita equivalence bimodule between $C(M)$ and the algebra $\Cl(M)$ of sections of the Clifford bundle of $M$.
It is natural to investigate if also the finite-dimensional spectral triple of the Standard Model $(A,H,D)$
 describes a (noncommutative) spin manifold, and in particular
if the elements of $H$ are in some sense ``spinors''. This condition
-- which we name ``property (M)'' in Def.~\ref{df:propM} --
can be precisely formulated again in terms of Morita equivalence involving $A$ and certain
noncommutative analogue of $\Cl(M)$, and is satisfied
in some basic examples like e.g.~Einstein-Yang-Mills systems.

We investigate the consequence of such a requirement on the finite non-commutative
geometry that should describe the Standard Model of elementary particles.
We shall show that in order to satisfy such a condition, we are forced to introduce two additional terms
in the Dirac operator, and consider a non-standard grading.
In order to get the correct experimental value of the Higgs mass, various modifications of the original model have been proposed:
to enlarge the Hilbert space thus introducing new fermions \cite{Ste09};
to turn one of the elements of the internal Dirac operator into a field by hand \cite{CC12}
rather than getting it as a fluctuation of the metric; 
to break (relax) the 1st order condition \cite{CCvS13a,CCvS13b},
thus allowing the presence of new terms in the Dirac operator;
to enlarge the algebra \cite{DLM13} and use a twisted spectral triple \cite{DM14} with bounded twisted commutators.
In the present paper from the Morita condition and a different grading we get two extra fields (without breaking any of the other conditions). 
We postpone to future work a discussion of the physical implications and, in particular, how the Higgs mass is modified.

Besides the original model, which is built around the real algebra:
\begin{equation}\label{eq:AF}
A_F=\C\oplus\Q\oplus M_3(\C) \;,
\end{equation}
where $\Q$ denotes the division ring of quaternions, we shall also consider the complex algebra 
\begin{equation}\label{eq:BF}
B_F=\C\oplus M_2(\C)\oplus M_3(\C) \;,
\end{equation}
which has an interesting interpretation from quantum group theory.
Namely, it is the semisimple part \cite{Coq97} of a certain quotient of
$U_q(sl(2))$ for $q$ a 3rd root of unity. As explained in \cite{DNS98},
the dual compact quantum group $Q$ fits into the exact sequence 
$$
1\to Q\to SL_q(2)\to SL(2,\C)\to 1 \;.
$$
Recall that $SL(2,\C)$ is a double covering of the restricted Lorentz group. One might argue
that trading a commutative space for an almost commutative one, the Lorentz group should be replaced
by a compact quantum group covering it, which takes into account the symmetries of the internal space
as well. (For preliminary studies of Hopf-algebra symmetries of $A_F$/$B_F$
see \cite{Kas95,Coq97,DNS98}; for compact quantum group symmetries see \cite{BBD10,BDDD11}.)

We show that a minimal modification in the representation allows to replace $A_F$ by $B_F$
without changing the content of the theory. In particular, at the representation
level the complexification $\pi(A_F)_{\C}$ of $\pi(A_F)$ is the minimal unitalization of the degenerate representation $\pi(B_F)$
(the representations, here denoted by the same symbol which  we will omit later on, are
introduced in \S\ref{sec:3});
adding the identity operator (which commutes with the Dirac operator) doesn't produce new fields.

\smallskip

The plan of the paper is the following. In \S\ref{sec:setup} we review some basic
ideas of noncommutative geometry \cite{Con94,GVF01,Lan02}, with a view to applications to gauge theory \cite{CM08,vS15}. In \S\ref{sec:3}, we review the
derivation of the finite spectral triple of the Standard Model and discuss an
alternative based on the complex algebra $B_F$ (\S\ref{sec:Hzero}). In \S\ref{sec:4},
we describe the most general Dirac operator satisfying the 1st order condition
(which is necessary for the ``property (M)'' in Def.~\ref{df:propM}),
and in \S\ref{sec:5} two possible grading operators; the Dirac operator of
Chamseddine-Connes \cite{CC97,Con06,CCM06,CM08} appears in \S\ref{sec:CC}.
In \S\ref{sec:6}, we discuss the natural condition for a spectral triple to be
``spin$^c$'', based on Morita equivalence, and derive some necessary conditions
for this to be satisfied; we show that in order to satisfy these condition one
has to introduce two additional terms in the Dirac operator of Chamseddine-Connes,
one mixing $e_R$ with $\bar\nu_R$ and one mixing leptons with quarks {(for a study of lepto-quarks in this setting, one can see \cite{PSS97})}. The last term
is also necessary in order to have an irreducible spectral triple, cf.~\S\ref{sec:8}.
In \S\ref{sec:orient}, we study the problem of orientability for the modified Dirac operator.
In \S\ref{sec:9} we discuss the irreducibility of the Pati-Salam model. We conclude in \S\ref{sec:10} with some final remarks.

\section{Mathematical set-up}\label{sec:setup}

Let $M$ be a closed oriented Riemannian manifold,
$C(M)$ and $C^\infty(M)$ the algebras of complex-valued continuous resp.~smooth functions,
and $\Cl(M)$ the algebra of continuous sections of the bundle of (complexified) Clifford algebras:
as a $C(M)$-module, it is equivalent to the module of continuous sections of the bundle $\Lambda^\bullet\hspace{1pt}T^*_{\C}M\to M$,
but with product defined by the Clifford multiplication.
The manifold $M$ is spin$^c$ if and only if there exists a Morita equivalence $\Cl(M)$-$C(M)$ bimodule $\Sigma$
(see e.g.~\S1 of \cite{Var06}). Such a $\Sigma$ is automatically projective and finitely generated, hence by
Serre-Swan theorem $\Sigma=\Gamma(S)$ is the module of sections of some complex vector bundle $S\to M$,
the \emph{spinor bundle} in the conventional picture from differential geometry.

Once we have $S$, we can introduce the Dirac operator $\D$, a self-adjoint operator on the Hilbert space
$L^2(M,S)$ of square integrable sections of $S\to M$ \cite[\S1.4]{Var06}. Let $\pi$ be the representation
of $C(M)$ on $L^2(M,S)$ by pointwise multiplication and $c$ the representation of $\Cl(M)$ by Clifford multiplication
(see e.g.~\cite{GVF01} or \cite{Var06} for the details).
The data
\begin{equation}\label{eq:can}
\big(\, C^\infty(M),\pi,L^2(M,S),\D \,\big)
\end{equation}
is the prototypical example of commutative spectral triple, and one can indeed prove under some additional assumptions
that any commutative spectral triple comes from such a construction \cite[Thm.~1.2]{Con08}.
The spectral triple \eqref{eq:can} is $\Z_2$-graded if $M$ is even dimensional.

There is an algebraic characterization for spin manifolds as well: a spin$^c$ manifold $M$ is spin if and only
if there exists a real structure for the spectral triple \eqref{eq:can} (whose definition we recall below in the
finite-dimensional case).

Let us observe that, for any $f\in C^\infty(M)$, $i[\D,\pi(f)]=c(\mathrm{d}f)$ is the operator of
Clifford multiplication by $\mathrm{d}f$ and such operators generate $\Cl(M)$.
In the even case, the grading $\gamma$ belongs to $\Cl(M)$.

\smallskip

For later use, we recall the definition of spectral triple in the finite-dimensional case, adapted to our purposes.

\begin{df}
A finite dimensional spectral triple $(A,\pi,H,D)$ 
is given by a finite dimensional complex Hilbert space $H$,
a Hermitian operator $D$ on $H$, and a real or complex $C^*$-algebra with a faithful $*$-representation
$\pi:A\to\mathrm{End}_{\C}(H)$.
The spectral triple is \emph{even} if $H$ is $\Z_2$-graded, $\pi(A)$ is even
and $D$ is odd; we denote by $\gamma$ the grading operator. The spectral triple is \emph{real} if there is
an antilinear isometry $J$ on $H$ -- called the \emph{real structure} -- satisfying
$$
J^2=\epsilon\,\id_H \;,\qquad\quad
JD=\epsilon' DJ \;,\qquad\quad
J\gamma=\epsilon''\gamma J \quad\text{(only in the even case)}
$$
for some $\epsilon,\epsilon',\epsilon''\in\{\pm 1\}$, together with the 0th order condition
$$
[\pi(a),J\pi(b)J^{-1}]=0 \qquad\forall\;a,b\in A,
$$
and the 1st order condition:
\begin{equation}\label{eq:1storder}
[[D,\pi(a)],J\pi(b)J^{-1}]=0 \qquad\forall\;a,b\in A.
\end{equation}
\end{df}

In order to simplify the notations, 
we will often omit the representation symbol $\pi$ and set $\gamma:=1$ if we have an odd spectral triple.
Note that we don't loose generality by assuming that the representation is faithful. Note also that $H$ is complex even when $A$ is real. 
The values of $\epsilon,\epsilon',\epsilon''$ determine the \emph{KO-dimension}
of the spectral triple (according to the table that is, for example, in \cite[\S3.8]{Var06}).

\begin{df}
Let $(A,H,D,\gamma)$ be a spectral triple (with $\gamma:=1$ if the spectral triple is odd) and $\Omega^1(A):=\mathrm{Span}\{a[D,b]:a,b\in A\}$.
We call $\Cl(A)_{\mathrm{o}}$ the complex $*$-algebra generated by $A$ and $\Omega^1$, and $\Cl(A)_{\mathrm{e}}$ the complex $*$-algebra generated by $\Cl(A)_{\mathrm{o}}$ and $\gamma$.
\end{df}

This is similar to Definition 3.19 of \cite{LRV12} ($\Cl(A)_{\mathrm{o}}$ is their $\mathcal{C_D(A)}$
in the even case, while in the odd case they double the Hilbert space to get a $\Z_2$-graded algebra).

\medskip

Let $A^\circ :=   JAJ^{-1}$ be the opposite algebra (thought of as a subalgebra of $\mathrm{End}_{\C}(H)$).
Recall that a linear map \mbox{$\pi_D:(A\otimes A^\circ)\otimes A^{\otimes n}\to\mathrm{End}_{\C}(H)$} is given by
$$
c=\sum_{\mathrm{finite}}(a_0^i\otimes b_0^i)\otimes a_1^i\otimes\ldots\otimes a_n^i\;\mapsto\;
\pi_{D}(c):=
\sum_{\mathrm{finite}}a_0^ib_0^i[D,a_1^i]
\ldots[D,a_n^i] \;,
$$
for all $a_j^i\in A$ and $b_0^i\in A^\circ$. 
By restriction (take $b_0^i=1$) we get a surjective map $$\pi_D:\bigoplus\nolimits_{n\geq 0} A^{\otimes n+1}\to \Cl(A)_{\mathrm{o}}$$ which we denote by the same symbol. Note that $\gamma$ is in the image of the latter map if and only if the two Clifford algebras coincide: $\Cl(A)_{\mathrm{e}}=\Cl(A)_{\mathrm{o}}$. This in particular happens when the spectral triple is orientable, cf.~below.

\begin{df}\label{df:OR}
Let $n\geq 0$. 
A spectral triple is \emph{orientable} (resp.~\emph{orientable in a weak sense}), with global dimension $\leq n$, if there exists a Hochschild cycle with coefficients in $A$ (resp.~in $A\otimes A^\circ$) such that $\pi_{D}(c)=\gamma$.
\end{df}

\noindent
Note that $c$ defines a class $[c]\in H\!H_n(A)$ (resp.~$[c]\in H_n(A,A\otimes A^\circ)$).
For a finite-dimensional real or complex $C^*$-algebra,
$H\!H_n(A)$ and $H_n(A,A\otimes A^\circ)$ are zero if $n>0$ (we thank U.~Kr{\"a}hmer for this remark). On the other hand, since $\pi_D$ is only defined on chains, 
rather than on homology classes ($[c]=0\not\Rightarrow\pi_D(c)=0$), it still makes sense to study orientability for arbitrary $n\geq 0$.

The 0th and 1st order conditions imply that $H$ is a $\Cl(A)_{\mathrm{e}}$-$A^\circ$ bimodule.
Indeed $a$ and $[D,a]$ commute with $b^\circ$ for all $a\in A,b^\circ\in A^\circ$, and $\gamma$ commutes with $A^\circ$ since it commutes with $A$ and $J\gamma=\epsilon''\gamma J$.
Inspired by the example \eqref{eq:can} we give then the following definition (much similar to the ``condition $5$'' of \cite{LRV12}):

\begin{df}\label{df:propM}
A spectral triple $(A,H,D,J,\gamma)$ has the \emph{property (M)} (resp.~\emph{property (M) with grading})
if $H$ is a Morita equivalence bimodule between $A^\circ$ and $\Cl(A)_{\mathrm{o}}$ (resp.~$\Cl(A)_{\mathrm{e}}$).
\end{df}

Since $\Cl(A)_{\mathrm{o}}\subset\Cl(A)_{\mathrm{e}}$, clearly the ``property (M) with grading'' is weaker. The two conditions are equivalent if the spectral triple is odd (so $\gamma=1$) or orientable.

\begin{ex}\label{ex}
If $H=A$, $J(a)=a^*$ and $D=0$, the spectral triple
has the property (M).
\end{ex}

\subsection{The gauge group of a real spectral triple}
Let $(A,\pi,H,D,J)$ be a real spectral triple,
and assume that $A$ is a unital and $\pi$ is a unital representation.
Let $\U(A)$ be the group of unitary elements of $A$. Due to the 0th order condition,
the map $\rho:\U(A)\to\mathrm{Aut}_{\C}(H)$ given by
\begin{equation}\label{eq:rhoa}
\rho(u):=uJuJ^{-1},
\end{equation}
is a representation, called \emph{adjoint representation}.

The gauge group $G(A)$ of a real spectral triple is defined as
$$
G(A):=\big\{uJuJ^{-1}:u\in\U(A)\big\} \;.
$$

\begin{ex}
In the spectral triple $(M_n(\C),M_n(\C),0,J)$ of
the Einstein Yang-Mills system the algebra acts by left multiplication, $J(a)=a^*$ is the Hermitian conjugation,
and the gauge group is $G(A)=\mathrm{PU}(n)$.
This spectral triple has the property (M), cf.~Example \ref{ex}.
\end{ex}

\subsection{Spectral triples with a degenerate representation}\label{sec:22}
A necessary and sufficient condition for the map $\rho$ in \eqref{eq:rhoa} to send $\U(A)$ into invertible
operators is that $\rho(1)=1$ (then automatically, $\rho(u^{-1})=\rho(u)^{-1}$).
A sufficient condition is that $\pi$ is a unital representation, that is $\pi(1)=1$.
For a spectral triple with a degenerate representation, the unit of $A$ is not the identity operator
on $H$, and \eqref{eq:rhoa} is in general not a representation of the unitary group $\U(A)$. Here we
explain how to bypass this problem.

Degenerate representations appear for example when one tries to sum a real spectral triple with one which has
no real structure.
Let $(A,\bar\pi_0,\bar H_0,0)$ and $(A,\pi_1,H_1,0,J_1)$ be two finite-dimensional
spectral triples, the latter one real with $J_1^2=1$, and both with the same algebra $A$ and null Dirac operator. Then we can define a new real
spectral triple $(A,\pi,H,0,J)$ as follows. We set
$$
H:=H_0\oplus\bar H_0\oplus H_1 \;,
$$
where $H_0=(\bar H_0)^*$ is the dual space.
We define
$$
\pi(a)(x,y,z)=\big(0,\bar\pi_0(a)y,\pi_1(a)z\big) \;,\qquad
J(x,y,z)=(y^*,x^*,J_1z) \;,
$$
for all $x\in H_0,y\in\bar H_0,z\in H_1$. Note that the representation $\pi$ is degenerate.
If we extend $\bar\pi_0$ and $\pi_1$ trivially to $H$ (as zero on $H_0\oplus H_1$ resp.~$H_0\oplus\bar H_0$),
then we can simply write:
$$
\pi=\bar\pi_0+\pi_1 \;.
$$
Since $\pi$ is degenerate, the map $u\mapsto\pi(u)J\pi(u)J$ is not a representation
of $\U(A)$ in $\mathrm{Aut}(H)$ (it doesn't map $1\mapsto 1$, and $u$ into an invertible operator).
A unitary representation $\rho$ of $U(A)$ on $H$ is given by
\begin{equation}\label{eq:reprho}
\rho(u):=\bar\pi_0(u)+J\bar\pi_0(u)J+\pi_1(u)J\pi_1(u)J \;.
\end{equation}
Indeed $\bar\pi_0(1)=\id_{\bar H_0}$, $\pi_1(1)=\id_{H_1}$ and
$J\bar\pi_0(1)J=\id_{H_0}$. So $\rho(1)=1$. Moreover,
$\bar\pi_0$, $J\bar\pi_0(\,.\,)J$, $\pi_1$ and $J\pi_1(\,.\,)J$
are mutually commuting, hence $\rho$ is multiplicative, and
from $\rho(u)\rho(u^*)=\rho(uu^*)=1$ we deduce that the representation
is also unitary.

Basically, we are considering the direct sum of three representations of $\U(A)$:
the fundamental associated to $\bar\pi_0$ and its dual, and the adjoint
representation of $\pi_1$.

In \S\ref{sec:Hzero} we exhibit a possible choice of the above data,
such that $\U(A)$ contains (strictly) the gauge group $G_{SM}$ of the Standard Model (modulo a finite subgroup), 
and $\rho|_{G_{SM}}$ gives the correct representation.


\section{From particles to algebras}\label{sec:3}

In this section we give a review of the derivation of the data $(A_F,H_F,J_F)$ from
physical considerations, and collect at the end few results about the algebra and its
commutant that will be useful in the following
sections. In some sense, these data reflect the ``topology'' of the finite noncommutative
manifold describing the internal space of the Standard Model, while the Dirac operator encodes
the metric properties. In \S\ref{sec:Hzero} we explain how to get the same gauge group from
a spectral triple based on the complex algebra \eqref{eq:BF} with a degenerate representation.
For simplicity, we work with only one generation of leptons/quarks.

\subsection{The gauge group of the Standard Model}
Let
$$
\widetilde{G}_{SM}:=\U(1)\times \SU(2)\times \SU(3)
$$
be the usual gauge group of the Standard Model, let $H$ be the finite-dimensional Hilbert space
representing the internal degrees of freedom of elementary fermions. Let us recall what is the
representation of $\widetilde{G}_{SM}$.
We have a decomposition $H=F\oplus F^*$, with $F^*$ the dual space of $F$.
The vector space $F$ (for \emph{fermions}) has basis
$$
\binom{\nu_L}{e_L} \qquad
\binom{u^c_L}{d^c_L}_{c=1,2,3} \;\;\qquad
\begin{matrix}
\nu_R &\quad & \{u^c_R\}_{c=1,2,3} \\ e_R && \{d^c_R\}_{c=1,2,3}
\end{matrix}
$$
where $\nu$ stands for neutrino, $e$ for electron, $u^c$ for up-quark and $d^c$
for down-quark with color $c=1,2,3$, $L,R$ stands for left-handed resp.~right-handed.
We will use the label $\uparrow$ for the first particle in each column (neutrino or quark up)
and $\downarrow$ for the second one (electron or quark down).
Left-handed doublets carry the fundamental representation of $\SU(2)$, while
right handed particles are $\SU(2)$-invariant; in particular, the $\uparrow$
particle in each doublet has \emph{weak isospin} $I_{3,w}=1/2$ and the $\downarrow$
has weak isospin $I_{3,w}=-1/2$. The $\SU(2)$-singlets have weak isospin $I_{3,w}=0$.
Each one of the color triplets carry the fundamental representation of $\SU(3)$, the
other particles being $\SU(3)$-invariant. Each particle carries a $1$-dimensional
representation $\lambda\to\lambda^{3Y_w}$ of $\U(1)$, where $Y_w\in\frac{1}{3}\Z$ is the
\emph{weak hypercharge}; it is computed from the formula $Q=I_{3,w}+\frac{1}{2}Y_w$ where
$Q$ is the charge of the particle. The value of $3Y_w$ is given by the following table:
\begin{center}\vspace{8pt}
\begin{tabular}{c|cccccc}
\hline
\textbf{particle} & $\nu_L , e_L$ & $u^c_L , d^c_L$ & $\nu_R$ & $e_R$ & $u^c_R$ & $d^c_R$ \\[2pt]
$3Y_w$ & $-3$ & $1$ & $0$ & $-6$ & $4$ & $-2$ \\
\hline
\end{tabular}
\end{center}\vspace{8pt}
The final representation is actually the direct sum of $n$ copies of $H = F \oplus F^*$,
where $n$ is the number of generations ($n=3$ according to our current knowledge).
For simplicity, the factor taking into account generations will be neglected.

For the computations, it will be convenient to 
encode the complex vector space $F$ of dimension $16$ as $F\simeq M_4(\C)$.
Namely we arrange the particles in a $4\times 4$ matrix as follows
$$
\begin{bmatrix}
\nu_R & u^1_R & u^2_R & u^3_R \\
e_R   & d^1_R & d^2_R & d^3_R \\
\nu_L & u^1_L & u^2_L & u^3_L \\
e_L   & d^1_L & d^2_L & d^3_L
\end{bmatrix} \;.
$$
We put in the first column leptons, in the other three the quarks according to the color.
In the rows we put in the order: $\uparrow R$, $\downarrow R$, $\uparrow L$, $\downarrow L$.

Let $e_{ij}$ the $4\times 4$ matrix with $1$ in position $(i,j)$ and zero everywhere else.
Matrices $\{e_{ij}\}_{i,j=1}^4$ form an orthonormal basis of $M_4(\C)$ for the inner product
associated to the trace $\inner{a,b}=\tr(a^*b)$. With this notation, for example, the state
associated to the unit vector $e_{31}$ represents a left handed neutrino.

In the dual representation $F^*$, one has:
$$
\begin{bmatrix}
\bar\nu_R & \bar e_R & \bar\nu_L & \bar e_L \\
\bar u^1_R & \bar{d}^{\,1}_R & \bar u^1_L & \bar{d}^{\,1}_L \\
\bar u^2_R & \bar{d}^{\,2}_R & \bar u^2_L & \bar{d}^{\,2}_L \\
\bar u^3_R & \bar{d}^{\,3}_R & \bar u^3_L & \bar{d}^{\,3}_L
\end{bmatrix} \;.
$$
Elements of $H$ are then of the form $a\oplus b$ with $a,b\in M_4(\C)$.

Endomorphisms of $F\simeq F^*\simeq M_4(\C)$ are given by $M_4(\C)\otimes M_4(\C)$, where the
first factor acts on $F=M_4(\C)$ via row-by-column multiplication from the left, and the second
via row-by-column multiplication from the right.
From the weak hypercharge table we get the following representation $\pi_{SM}$ of $\widetilde{G}_{SM}$
on $H$:
\begin{multline*}
\pi_{SM}(\lambda,q,m)=\left[\!
\begin{array}{c|c}
\begin{matrix} \;\lambda^3\; & \;0\;  \\ 0 & \bar\lambda^3 \end{matrix}  &
\begin{matrix} \;0\; & \;0\;  \\ 0 & 0 \end{matrix} \\
\hline
\begin{matrix} \;0\; & \;0\; \\ 0 & 0 \end{matrix} & q
\end{array}
\!\right]\otimes
\left[\!
\begin{array}{c|ccc}
\bar\lambda^3 & \;0\; & \;0\; & \;0\; \\
\hline
\begin{matrix} \;0\; \\ 0 \\ 0 \end{matrix} && \lambda m^*
\end{array}
\!\right]
\oplus
\\[5pt]
\oplus
\left[\!
\begin{array}{c|ccc}
\lambda^3 & \;0\; & \;0\; & \;0\; \\
\hline
\begin{matrix} \;0\; \\ 0 \\ 0 \end{matrix} && \bar\lambda m
\end{array}
\!\right]\otimes
\left[\!
\begin{array}{c|c}
\begin{matrix} \;\bar\lambda^3\; & \;0\;  \\ 0 & \lambda^3 \end{matrix}  &
\begin{matrix} \;0\; & \;0\;  \\ 0 & 0 \end{matrix} \\
\hline
\begin{matrix} \;0\; & \;0\; \\ 0 & 0 \end{matrix} & q^*
\end{array}
\!\right]
\end{multline*}
for all $\lambda\in\U(1)$, $q\in\SU(2)$ and $m\in\SU(3)$. Here the first summand acts on $F$ and the
second one on $F^*$.

Computing the kernel of $\pi_{SM}$ one sees that the relevant group is not exactly $\widetilde{G}_{SM}$ but
a quotient by a finite subgroup. Let $\Z_6:=\{\mu\in\C:\mu^6=1\}$ the group of $6$-th roots of unity.
There is an injective morphism of groups:
\begin{equation}\label{eq:Z6}
\Z_6\ni\mu\mapsto (\mu,\mu^3 1_2,\mu^4 1_3)\in \widetilde{G}_{SM} \;.
\end{equation}
One easily checks that the image is exactly the kernel of $\pi_{SM}$, so that there is an exact sequence
$$
1\to\Z_6=\ker\pi_{SM}\to \widetilde{G}_{SM}\to\mathrm{Im}\,\pi_{SM}\to 1
$$
We want to identify $\mathrm{Im}\,\pi_{SM}=\widetilde{G}_{SM}/\Z_6$. Let $G_{SM}:=\mathrm{S}(\U(2)\times\U(3))$
be the group of $\SU(5)$ matrices of the form{\small
$$
\begin{bmatrix}
2\times 2 \text{ block} & 0 \\
0 & 3\times 3 \text{ block}
\end{bmatrix}
$$
}There is a surjective morphism of groups
\begin{equation}\label{eq:insp}
\widetilde{G}_{SM}\ni (\lambda,q,m)\mapsto 
\begin{bmatrix}
\lambda^3q \\ &\bar\lambda^2\bar m
\end{bmatrix}
\in G_{SM} \;,
\end{equation}
where $\bar m=(m^*)^t$.
An element $(\lambda,q,m)$ is in the kernel of the map $\widetilde{G}_{SM}\to G_{SM}$ if and only if
$q=\bar\lambda^31_2$ and $m=\bar\lambda^21_3$. But $\det(q)=\det(m)=\bar\lambda^6$ must be
$1$, hence $\lambda\in\Z_6$ and
$$
\mathrm{Im}\,\pi_{SM}\simeq G_{SM} \;.
$$
The representation can be linearized as follows. 
Let $J:H\to H$ be the antilinear operator $J(a\oplus b):=b^*\oplus a^*$, transforming a particle
into its antiparticle. We can write
$$
\pi_{SM}(\lambda,q,m)=\tilde\pi(\tilde\lambda,q,\tilde m)\,J\tilde\pi(\tilde\lambda,q,\tilde m)J^{-1}
$$
where $\tilde m:=\bar\lambda m\in\U(3)$, $\tilde\lambda:=\lambda^3$ and
$$
\tilde\pi(\tilde\lambda,q,\tilde m)=\left[\!
\begin{array}{c|c}
\begin{matrix} \;\tilde\lambda\; & \;0\;  \\ 0 & \tilde\lambda^* \end{matrix}  &
\begin{matrix} \;0\; & \;0\;  \\ 0 & 0 \end{matrix} \\
\hline
\begin{matrix} \;0\; & \;0\; \\ 0 & 0 \end{matrix} & q
\end{array}
\!\right]\otimes 1
\;\oplus\;
\left[\!
\begin{array}{c|ccc}
\tilde\lambda & \;0\; & \;0\; & \;0\; \\
\hline
\begin{matrix} \;0\; \\ 0 \\ 0 \end{matrix} && \tilde m
\end{array}
\!\right]\otimes 1
$$
The latter can be now extended in an obvious way, by $\R$-linearity, as
a representation of the real algebra $A_F$ in \eqref{eq:AF},
where we think of quaternions as matrices in $M_2(\C)$  of the form
$$
\left[\!\begin{array}{rr} \alpha & \beta \\ \!-\bar\beta & \;\bar\alpha \end{array}\!\right]
\;,\qquad\alpha,\beta\in\C,
$$
so that with this identification $\U(\Q)=\SU(2)$.

\subsection{The data $(A_F,A_F^\circ,H_F,J_F)$}\label{sec:thedata}
With the identifications as in the previous section the Hilbert space becomes $H_F=M_{8\times 4}(\C)$, with elements:
$$
v=\begin{bmatrix} v_1 \\ v_2 \end{bmatrix} \;,\qquad v_1,v_2 \in M_4(\C).
$$
and inner product $\inner{v,w}=\tr(v^*w)$.
Linear operators on $H_F$ are finite sums $L=\sum_i a_i\otimes b_i$, with $a_i\in M_8(\C)$
acting via row-by-column multiplication from the left and $b_i\in M_4(\C)$ acting
via row-by-column multiplication from the right. One easily checks that the adjoint of $L$ is
$L^*=\sum_ia_i^*\otimes b_i^*$, with $a_i^*,b_i^*$ denoting Hermitian conjugation.

The real structure $J_F$ is the operator
\begin{equation}\label{eq:JF}
J_F\begin{bmatrix} v_1 \\ v_2 \end{bmatrix}=\begin{bmatrix} v_2^* \\ v_1^* \end{bmatrix} \;.
\end{equation}
We identify $A_F=\C\oplus\Q\oplus M_3(\C)$ with the subalgebra of elements $a\otimes 1\in\mathrm{End}_{\C}(H_F)$,
with $a$ of the form:
\begin{equation}\label{eq:8t8}
a=\begin{bmatrix}
\left[\!
\begin{array}{c|c}
\begin{matrix} \;\lambda\; & \;0\;  \\ 0 & \bar\lambda \end{matrix} &
\begin{matrix} \;0\; & \;0\;  \\ 0 & 0 \end{matrix} \\
\hline
\begin{matrix} \;0\; & \;0\; \\ 0 & 0 \end{matrix} & q
\end{array}
\!\right] \!\! \\ & \!\!
\left[\!
\begin{array}{c|ccc}
\lambda & \;0\; & \;0\; & \;0\; \\
\hline
\begin{matrix} \;0\; \\ 0 \\ 0 \end{matrix} && m
\end{array}
\!\right]
\end{bmatrix} \;,
\end{equation}
with $\lambda\in\C$, $q\in\Q$ and $m\in M_3(\C)$ (with zeros on the off-diagonal blocks).

We denote by $A_F^\circ=J_FA_FJ_F$ the subalgebra of elements $\mathrm{End}_{\C}(H_F)$
of the form:
$$
a^\circ=
\begin{bmatrix} 1 & 0 \\ \;0\; & \;0\; \end{bmatrix}\otimes
\left[\!
\begin{array}{c|ccc}
\lambda & \;0\; & \;0\; & \;0\; \\
\hline
\begin{matrix} \;0\; \\ 0 \\ 0 \end{matrix} && m
\end{array}
\!\right]
+\begin{bmatrix} 0 & \;0\; \\ \;0\; & 1 \end{bmatrix}\otimes\left[\!
\begin{array}{c|c}
\begin{matrix} \;\lambda\; & \;0\;  \\ 0 & \bar\lambda \end{matrix}  &
\begin{matrix} \;0\; & \;0\;  \\ 0 & 0 \end{matrix} \\
\hline
\begin{matrix} \;0\; & \;0\; \\ 0 & 0 \end{matrix} & q
\end{array}
\!\right] \;.
$$
(On the first factor of each tensor $0,1\in M_4(\C)$ are the zero and the identity matrix.)

If $A\subset\mathrm{End}_{\C}(H_F)$ is a real $*$-subalgebra, we denote
by $A_{\C}$ the complex linear span of the elements in $A$; note that $A$ and $A_{\C}$ have
the same commutant in $\mathrm{End}_{\C}(H_F)$. The map $a\mapsto a^\circ=J_F\bar aJ_F$ (here $\bar a=(a^*)^t$)
gives two isomorphisms $A_F\to A_F^\circ$ and $(A_F)_{\C}\to (A_F^\circ)_{\C}$.

\begin{lemma}\label{lemma:1}
The commutant of the algebra of elements \eqref{eq:8t8} in $M_8(\C)$ is the algebra $C_F$ with elements\vspace{10pt}
\addtolength{\arraycolsep}{3pt}
\begin{equation}\label{eq:CF}
\left[\begin{array}{ccccc}
\multicolumn{1}{c|}{\rule{0pt}{15pt}q_{11}} &&& \multicolumn{1}{|c|}{\!\rule{0pt}{15pt}q_{12}} \\[5pt]
\cline{1-2}\cline{4-4}
& \multicolumn{1}{|c|}{\rule{0pt}{15pt}\,\alpha\,} \\[4pt]
\cline{2-3}
&& \multicolumn{1}{|c|}{\;\;\beta 1_2\rule[-15pt]{0pt}{37pt}\;} \\
\cline{1-1}\cline{3-4}
\multicolumn{1}{c|}{\rule{0pt}{15pt}q_{21}} &&& \multicolumn{1}{|c|}{\rule{0pt}{15pt}q_{22}} \\[5pt]
\cline{1-1}\cline{4-5}
\begin{matrix} ~ \\ ~ \\ ~ \end{matrix}
&&&& \multicolumn{1}{|c}{\quad\;\delta 1_3\quad\rule[-16pt]{0pt}{37pt}}
\end{array}\right] \;,
\end{equation}
where the $\beta$-block is $2\times 2$, the $\delta$-block is $3\times 3$, and all other framed blocks are $1\times 1$
($\alpha,\beta,\delta\in\C$, $q=(q_{ij})\in M_2(\C)$). All other blocks are zero (zeroes are omitted).

The commutant of $A_F$ in $\mathrm{End}_{\C}(H)$ is $A_F'=C_F\otimes M_4(\C)$.
\end{lemma}

\begin{proof}
By direct computation.
\end{proof}

\noindent
Note that $A_F'\simeq M_4(\C)^{\oplus 3}\oplus M_8(\C)$.
The map $x\mapsto J_F\bar xJ_F$ is an isomorphism between $A_F'$ and $(A_F^\circ)'$.
From this, we get the following result.

\begin{lemma}\label{lemma:2}
The commutant $(A_F^\circ)'$ of $A_F^\circ$ has elements
\begin{equation}\label{eq:lemma2}
a\otimes e_{11}+
\begin{bmatrix}
\;b\; \\ & \;c\;
\end{bmatrix}\otimes e_{22}
+\begin{bmatrix}
\;b\; \\ & \;d\;
\end{bmatrix}\otimes (e_{33}+e_{44})
\end{equation}
with $a\in M_8(\C)$, $b,c,d\in M_4(\C)$.
\end{lemma}


\subsection{The data $(B_F,\bar H_0,H_1)$}\label{sec:Hzero}
In this section we explain how to get the same gauge group from
a spectral triple based on the complex algebra \eqref{eq:BF} with a degenerate representation.

Let us put particles into a row vector and a $3\times 4$ matrix as follows
$$
\begin{bmatrix}
e_R  & d^1_R & d^2_R & d^3_R 
\end{bmatrix}
\;,\qquad
\begin{bmatrix}
\nu_R & u^1_R & u^2_R & u^3_R \\
\nu_L & u^1_L & u^2_L & u^3_L \\
e_L   & d^1_L & d^2_L & d^3_L
\end{bmatrix} \;,
$$
Thus a particle is represented by a vector in
$\C^3\oplus M_{3\times 4}(\C)$, with inner product given on each summand by $\inner{v,w}=\tr(v^*w)$.
Antiparticles belong to the dual space.

The Hilbert space is then $H=H_0\oplus \bar H_0\oplus H_1$, where elements of
$H_0\simeq\C^4$ are row vectors, elements of
$\bar H_0\simeq\C^4$ are column vectors, and $H_1$ has elements
$$
\sbinom{a}{b} \;,\qquad a\in M_{3\times 4}(\C),b \in M_{4\times 3}(\C).
$$
The real structure $J$ is the antilinear operator:
\begin{equation}\label{eq:JBF}
J\Big(v\oplus w\oplus\sbinom{a}{b}\Big)=
w^*\oplus v^*\oplus\sbinom{b^*}{a^*} \;.
\end{equation}
We define two unital $*$-representations
$\bar\pi_0:B_F\to\mathrm{End}_{\C}(\bar H_0)$
and $\pi_1:B_F\to\mathrm{End}_{\C}(H_1)$
of the algebra $B_F=\C\oplus M_2(\C)\oplus M_3(\C)$ in \eqref{eq:BF}
as follows
$$
\bar\pi_0(\lambda,q,m)=\left[\!
\begin{array}{c|ccc}
\lambda & \;0\; & \;0\; & \;0\; \\
\hline
\begin{matrix} \;0\; \\ 0 \\ 0 \end{matrix} && m
\end{array}
\!\right]
,\quad
\pi_1(\lambda,q,m)=\begin{bmatrix}
\left[\!
\begin{array}{c|c}
\;\lambda\; &
\begin{matrix} \;0\; & \;0\; \end{matrix} \\
\hline
\begin{matrix} \;0\; \\ 0 \end{matrix} & q
\end{array}
\!\right] \!\!\! \\ & \!\!\!
\left[\!
\begin{array}{c|ccc}
\lambda & \;0\; & \;0\; & \;0\; \\
\hline
\begin{matrix} \;0\; \\ 0 \\ 0 \end{matrix} && m
\end{array}
\!\right]
\end{bmatrix},
$$
both acting via row-by-column multiplication from the left.
Here $\lambda\in\C$, $q\in M_2(\C)$ and $m\in M_3(\C)$ are $2\times 2$ and $3\times 3$ blocks,
and the off-diagonal $3\times 4$ and $4\times 3$ blocks are zero.

An (injective) representation $\rho$ of $\U(A)=\U(1)\times\U(2)\times\U(3)$ is given by \eqref{eq:reprho}.

One can check that $\rho$, composed with the map
$$
\widetilde{G}_{SM}=\U(1)\times\SU(2)\times\SU(3)\xrightarrow{\varphi} \U(A) \;,\qquad
(\lambda,q,m)\mapsto (\lambda^6,\lambda^3q,\lambda^2m) \;,
$$
gives the correct representation of $\widetilde{G}_{SM}$ (in particular,
each particle has the correct weak hypercharge). The kernel $\varphi$ is given again
by the elements in \eqref{eq:Z6}, so that the range of $\varphi$ is $G_{SM}\simeq\widetilde G_{SM}/\Z_6$.
The map
$$
\U(A)\supset G_{SM}\ni (\lambda,q,m)\mapsto \begin{bmatrix}
q \\ & \bar m
\end{bmatrix} \in\mathrm{S}(\U(2)\times\U(3))
$$
is an isomorphism. We then recover $G_{SM}$ as the subgroup of $\U(A)$ satisfying
the unimodularity condition
$$
\det\bar\pi_0(u)=\det\pi_1(u)=1 \;.
$$
The relation with $A_F$ is as follows.
Let $\{a_i\}_{i=1}^4$ be the rows of $a\in M_4(\C)$ and $\{b_j\}_{j=1}^4$ the columns of $b\in M_4(\C)$.
With the isometry
$$
H\ni a_2\oplus b_2\oplus
\left[\begin{array}{ccc}\multicolumn{3}{c}{a_1} \\ \multicolumn{3}{c}{a_3} \\ \multicolumn{3}{c}{a_4} \\
\hline
 b_1 & b_3 & b_4 \end{array}\right]\longrightarrow
\sbinom{a}{b}
\in M_{8\times 4}(\C)
$$
we transform $J$ in \eqref{eq:JBF} into the real structure $J_F$ in \eqref{eq:JF},
and $\pi$ into the representation (denoted by the same symbols):
\begin{equation}\label{eq:7t7}
\pi(\lambda,q,m)=\begin{bmatrix}
\left[\!
\begin{array}{c|c}
\begin{matrix} \;\lambda\; & \;0\;  \\ 0 & 0 \end{matrix} &
\begin{matrix} \;0\; & \;0\;  \\ 0 & 0 \end{matrix} \\
\hline
\begin{matrix} \;0\; & \;0\; \\ 0 & 0 \end{matrix} & q
\end{array}
\!\right] \!\! \\ & \!\!
\left[\!
\begin{array}{c|ccc}
\lambda & \;0\; & \;0\; & \;0\; \\
\hline
\begin{matrix} \;0\; \\ 0 \\ 0 \end{matrix} && m
\end{array}
\!\right]
\end{bmatrix}\otimes 1 \;,
\end{equation}
where $\lambda\in\C$, $q\in M_2(\C)$ and $m\in M_3(\C)$.

Note that the only difference between the matrix in \eqref{eq:7t7} and the one in \eqref{eq:8t8}
is the zero in position $(2,2)$ replacing $\bar\lambda$.
More precisely, the algebra $(A_F)_{\C}$ is the minimal unitalization of $\pi(B_F)$
in $\mathrm{End}_{\C}(H_F)$, and $(A_F^\circ)_{\C}$ is the unitalization of $A^\circ:=J_F\pi(B_F)J_F$.

Adding the identity doesn't change the commutant, nor $\Omega^1$. Thus, the results in the
next section which we state for the algebra $A_F$ are valid for $B_F$ as well.


\section{The 1st order condition}\label{sec:4}

In this section, we describe the most general Dirac operator satisfying the 1st order condition, which is the crucial one for a study of the property (M).
To keep things general, at the beginning we make no assumption regarding the other axioms (parity, KO-dimension, etc.).
We will then impose the additional requirement $J_FD_F=D_FJ_F$, with the plus sign on the right hand side dictated by the physical
content of the theory (the mass terms in the spectral action come from elements commuting with $J_F$).
It turns out that for any $D_F$ satisfying the 1st order condition, there is one commuting with $J_F$ which gives the same Clifford algebra (so, the condition $J_FD_F=D_FJ_F$ does not create any particular problem).
In the next sections, we will discuss the issue of the grading and the property (M), with or without grading.

The next proposition was originally stated in \cite[\S3.4]{Kra97}, and proved by
decomposing the $A$-bimodule $H$ into
irreducible ones and determining the corresponding matrix elements of $D$. Here, without assuming the orientability condition, we present an alternative proof that doesn't make use of such a decomposition.

\begin{prop}\label{prop:1}
Let $H$ be an $A\otimes A^\circ$-bimodule (i.e.~$[a,b^\circ]=0\;\forall\;a\in A$ and $b^\circ\in A^\circ$).
Then $D\in\mathrm{End}_{\C}(H)$ satisfies the 1st order condition --- i.e.~$[[D,a],b^\circ]=0\;\forall\;a\in A,b^\circ\in A^\circ$ --- if and only if it is of the form
$$
D=D_0+D_1
$$
where $D_0\in (A^\circ)'$ and $D_1\in A'$.
\end{prop}

We need a preliminary Lemma.

\begin{lemma}\label{lemma:2b}
Let $H$ be finite-dimensional and $V$ any $*$-subalgebra of $\mathrm{End}(H)$.
Then, there exists a direct complement $W$ of $V$
in $\mathrm{End}(H)$
satisfying $[V,W]\subseteq W$.
\end{lemma}
\begin{proof}
Modulo an isomorphism, we can assume $H=\C^n$ for some $n$, and $\mathrm{End}(H)=M_n(\C)$.
Let $W=V^\perp$ be the orthogonal complement with respect to the Hilbert-Schmidt inner product: $\inner{v,w}_{\mathrm{HS}}:=\tr(v^*w)\;\forall\;v,w\in M_n(\C)$.
For all $a,b\in V$ and $c\in V^\perp$, using the cyclic property of the trace, we derive:
$$
\inner{a,[b,c]}_{\mathrm{HS}}=
\inner{[b^*,a],c}_{\mathrm{HS}}=0 \;,
$$
where in last step we noticed that $[b^*,a]\in V$,
since $V$ is a $*$-algebra, and then the inner product is zero.
Thus $[b,c]\in V^\perp$, and $[V,V^\perp]\subset V^\perp$.
\end{proof}

\begin{proof}[Proof of Prop.~\ref{prop:1}]
The ``if'' part is trivial; we now prove the ``only if''. We want to prove that the 1st order condition implies $D\in A'+(A^\circ)'$, where by $A'+(A^\circ)'$ we mean the vector space generated by the commutants $A'$ and $(A^\circ)'$ (not the algebra).

We apply Lemma \ref{lemma:2b} to $V=A'$ and decompose $D=D_0+D_1$ with $D_0\in A'$ and $D_1\in W$.
From the 1st order condition:
$$
[a,[D_1,b^\circ]]=
[[D_1,a],b^\circ]-[D_1,[a,b^\circ]]=
[[D,a],b^\circ]+0=0 \;,
$$
for all $a\in A$, $b^\circ\in A^\circ$.
Hence
$[D_1,b^\circ]\in A'$
for all $b^\circ\in A^\circ$. But $A^\circ\subseteq A'$ 
and from Lemma \ref{lemma:2b} we also have $[D_1,b^\circ]\in W$. Since the sum $\mathrm{End}(H)=A'\oplus W$ is direct, it must be $[D_1,b^\circ]=0$.
This means that $D_1\in (A^\circ)'$.
\end{proof}

Note that, contrary to \cite{Kra97}, here the decomposition in Prop.~\ref{prop:1} is not necessarily unique. Uniqueness of the decomposition in \cite{Kra97} follows from the orientability condition. However, we'll see that in the Standard Model example, the spectral triple is not orientable (cf.~\S\ref{sec:orient}) and $A_F'\cap (A_F^\circ)'$ is not zero.

We now come back to the Standard Model.
In the rest of the paper,
we employ  $A_F,A_F^\circ,H_F, J_F$ as defined 
in \S\ref{sec:thedata},
but the same results are valid for the algebra $B_F$ and the representation discussed in \S\ref{sec:Hzero}.

\smallskip

\begin{prop}\label{prop:2}
An operator $D_F=D_F^*$ as in Prop.~\ref{prop:1} commutes with $J_F$ if and only if it is of the form:
$$
D_F=D_0+J_FD_0J_F
$$
with $D_0=D_0^*\in (A_F^\circ)'$.
\end{prop}
\begin{proof}
$x\mapsto J_F\bar xJ_F$ gives a bijection $A_F'\to (A_F^\circ)'$. The condition $J_FD_FJ_F=D_F$
gives
$$
(J_FD_0J_F-D_1)+(J_FD_1J_F-D_0)=0 \;.
$$
Since the first term is in $A_F'$ and the second in $(A_F^\circ)'$, the sum is zero if and only if
both $J_FD_0J_F-D_1$ and $J_FD_1J_F-D_0$ belong to $A_F'\cap (A_F^\circ)'$.
Called $D'=D_1-J_FD_0J_F$, one has the decomposition
$$
D_F=D_0+J_FD_0J_F+D' \;.
$$
From $J_FD_FJ_F-D_F=J_FD'J_F-D'$ one deduces that $J_F$ and $D'$ must commute.
So $D_F=(D_0+D'/2)+J_F(D_0+D'/2)J_F$ and we get the decomposition \eqref{eq:decom},
after renaming $D_0+D'/2\to D_0$.

Decompose $D_0=S+iT$ with $S$ and $T$ selfadjoint. Since $J_F$ is antilinear and $J_F=J_F^*$:
$$
D_F-D_F^*=2i(T-J_FTJ_F)
$$
which must be zero. But this implies
$$
D_F=S+J_FSJ_F+i(T-J_FTJ_F)=S+J_FSJ_F \;.
$$
Renaming $S\leadsto D_0$ (which now is selfadjoint) we conclude the proof.
\end{proof}

\begin{rem}\label{rem:12}
Note that the $D_1$ term does not contribute to $\Cl(A_F)$ (it commutes with $A_F$).
Then, for any Dirac operator as in Prop.~\ref{prop:1}, we can find one commuting with $J_F$
(replacing $D_1$ by $J_FD_0J_F$) without changing the Clifford algebra $\Cl(A_F)$.
In particular, the property (M) puts constrains only on $D_0$.
\end{rem}

It is useful to reformulate Prop.~\ref{prop:1} and Prop.~\ref{prop:2} as follows. Let
\begin{equation}\label{eq:UpsilonR}
D_R:=(\Upsilon_Re_{51}+\bar\Upsilon_Re_{15})\otimes e_{11} \;,
\end{equation}
with $\Upsilon_R\in\C$. Note that $D_R\in A_F'\cap (A_F^\circ)'$ and $J_FD_R=D_RJ_F$.

\begin{prop}\label{prop:9}
The most general $D_F=D_F^*$ satisfying the 1st order condition is 
\begin{equation}\label{eq:decom}
D_F=D_0+D_1+D_R
\end{equation}
where $D_0=D_0^*\in (A_F^\circ)'$ and $D_1=D_1^*\in A_F'$ have null entry in direction
of $e_{15}\otimes e_{11}$ and $e_{51}\otimes e_{11}$, and $D_1=J_FD_0J_F$ if $D_F$ and $J_F$ commute.
\end{prop}

In this way we isolated all the terms which do not contribute to $\Omega^1$. For any $a\in A_F$
and $D_F$ as in \eqref{eq:decom}, $[D_F,a]=[D_0,a]$.

\section{The grading operator}\label{sec:5}

\begin{lemma}\label{lemma:8}
Let $\gamma_F$ be a grading operator. Any odd Dirac operator satisfying
the 1st order condition can be written in the form $D_F=D_0+D_1+\kappa D_R$ as in Prop.~\ref{prop:9}, with both $D_0$ and
$D_1$ odd operators and $\kappa=0$ or $1$ depending on the parity of $D_R$.
\end{lemma}

\begin{proof}
From Prop.~\ref{prop:1}, we can write $D_F=D_0+D_1+T_0+T_1$ where $D_0,T_0\in (A_F^\circ)'$, $D_1,T_1\in A_F'$,
$D_0,D_1$ are odd and $T_0,T_1$ are even.
From
$$
\gamma_FD_F\gamma_F+D_F=2(T_0+T_1)=0
$$
we deduce $T_0+T_1=0$, so that $D_F=D_0+D_1$ with both $D_0$ and $D_1$ odd operators.
\end{proof}

\begin{lemma}\label{lemma:9}
Let $\gamma_F$ be a grading operator either commuting or anticommuting with $J_F$. Any odd Dirac operator
satisfying the 1st order condition and commuting with $J_F$ can be written in the form $D_F=D_0+J_FD_0J_F
+\kappa D_R$ as in Prop.~\ref{prop:9}, with $D_0$ an odd operator and $\kappa=0$ or $1$ depending on
the parity of $D_R$.
\end{lemma}

\begin{proof}
It follows from Lemma \ref{lemma:8}.
Since $D_1=J_FD_0J_F$, the condition $\gamma_FD_0\gamma_F=-D_0$ implies $\gamma_FD_1\gamma_F=-D_1$.
\end{proof}

\noindent
We now study the form of $D_0$ for Dirac operators of the type described by Lemma \ref{lemma:8} or \ref{lemma:9}
for two natural choices of the grading operator (we ignore $D_1$, cf.~Remark~\ref{rem:12}). It is worth noticing that both such gradings anticommute with $J_F$, and then give real spectral triples of KO-dimension $6$.

\subsection{The standard grading}\label{sec:Sgamma}
The grading in \cite{CC97,Con06,CCM06,CM08} (the chirality operator) is:
\begin{equation}\label{eq:Sgamma}
\gamma_F=
\begin{bmatrix} 1_2 \\ & \!\!-1_2 \\ && 0_4 \end{bmatrix}
\otimes 1_4
+
\begin{bmatrix} 0_4 \\ & \!\!-1_4 \end{bmatrix}
\otimes
\begin{bmatrix} 1_2 \\ & \!\!-1_2 \end{bmatrix} \;.
\end{equation}
It follows from Lemma \ref{lemma:2} that any $D_0$ anticommuting with \eqref{eq:Sgamma} has the form:
$$
D_0=\text{\footnotesize$\begin{bmatrix}
\zero & \zero & * & * & \zero & \circledast & \circledast & \circledast \\
\zero & \zero & * & * & * & \circledast & \circledast & \circledast \\
* & * & \zero & \zero & \zero & \zero & \zero & \zero \\
* & * & \zero & \zero & \zero & \zero & \zero & \zero \\
\zero & * & \zero & \zero & \zero & \zero & \zero & \zero \\
\circledast & \circledast & \zero & \zero & \zero & \zero & \zero & \zero \\
\circledast & \circledast & \zero & \zero & \zero & \zero & \zero & \zero \\
\circledast & \circledast & \zero & \zero & \zero & \zero & \zero & \zero
\end{bmatrix}$}\otimes e_{11}
+\text{\footnotesize$\begin{bmatrix}
\zero & \zero & * & * & \zero & \zero & \zero & \zero \\
\zero & \zero & * & * & \zero & \zero & \zero & \zero \\
* & * & \zero & \zero & \zero & \zero & \zero & \zero \\
* & * & \zero & \zero & \zero & \zero & \zero & \zero \\
\zero & \zero & \zero & \zero & \zero & \zero & \zero & \zero \\
\zero & \zero & \zero & \zero & \zero & \zero & \zero & \zero \\
\zero & \zero & \zero & \zero & \zero & \zero & \zero & \zero \\
\zero & \zero & \zero & \zero & \zero & \zero & \zero & \zero
\end{bmatrix}$}\otimes (1-e_{11}) \;,\hspace*{-5pt}
$$
where the asterisks indicate the only positions where one can have non-zero matrix entries.
The circled entries ($\circledast$) are the ones
that are not allowed by the non-standard grading \eqref{eq:Ngamma}.

\subsection{A non-standard grading}\label{sec:Ngamma}
Let
\begin{equation}\label{eq:Ngamma}
\gamma_F=
\begin{bmatrix} 1_2 \\ & \!\!-1_2 \\ && 0_4 \end{bmatrix}
\otimes
\begin{bmatrix} 1 \\ & \!\!-1_3 \end{bmatrix}
+
\begin{bmatrix} 0_4 \\ & \!\!-1 \\ && 1_3 \end{bmatrix}
\otimes
\begin{bmatrix} 1_2 \\ & \!\!-1_2 \end{bmatrix} \;.
\end{equation}
This operator assigns opposite parity to chiral leptons and quarks
(left resp.~right handed leptons have the same parity of right resp.~left handed quarks).

Again from Lemma \ref{lemma:2}, any $D_0$ anticommuting with \eqref{eq:Ngamma} has the form:
\begin{align*}
D_0 = &\text{\footnotesize$\left[\begin{matrix}
\zero & \zero & * & * & \multicolumn{1}{|c}{\zero} & \zero & \zero & \zero \\
\zero & \zero & * & * & \multicolumn{1}{|c}{*} & \zero & \zero & \zero \\
* & * & \zero & \zero & \multicolumn{1}{|c}{\zero} & \circledast & \circledast & \circledast \\
* & * & \zero & \zero & \multicolumn{1}{|c}{\zero} & \circledast & \circledast & \circledast \\
\cline{5-8}
\zero & * & \zero & \zero & \zero & \circledast & \circledast & \circledast \\
\zero & \zero & \circledast & \circledast & \circledast & \zero & \zero & \zero \\
\zero & \zero & \circledast & \circledast & \circledast & \zero & \zero & \zero \\
\zero & \zero & \circledast & \circledast & \circledast & \zero & \zero & \zero
\end{matrix}\,\,\right]$}\otimes e_{11}
+\text{\footnotesize$\begin{bmatrix}
\zero & \zero & * & * & \zero & \zero & \zero & \zero \\
\zero & \zero & * & * & \zero & \zero & \zero & \zero \\
* & * & \zero & \zero & \zero & \zero & \zero & \zero \\
* & * & \zero & \zero & \zero & \zero & \zero & \zero \\
\zero & \zero & \zero & \zero & \zero & \zero & \zero & \zero \\
\zero & \zero & \zero & \zero & \zero & \zero & \zero & \zero \\
\zero & \zero & \zero & \zero & \zero & \zero & \zero & \zero \\
\zero & \zero & \zero & \zero & \zero & \zero & \zero & \zero
\end{bmatrix}$}\otimes (1-e_{11})
\\
+ &\text{\footnotesize$\begin{bmatrix}
\zero & \zero & \zero & \zero & \zero & \zero & \zero & \zero \\
\zero & \zero & \zero & \zero & \zero & \zero & \zero & \zero \\
\zero & \zero & \zero & \zero & \zero & \zero & \zero & \zero \\
\zero & \zero & \zero & \zero & \zero & \zero & \zero & \zero \\
\zero & \zero & \zero & \zero & \zero & \circledast & \circledast & \circledast \\
\zero & \zero & \zero & \zero & \circledast & \zero & \zero & \zero \\
\zero & \zero & \zero & \zero & \circledast & \zero & \zero & \zero \\
\zero & \zero & \zero & \zero & \circledast & \zero & \zero & \zero
\end{bmatrix}$}\otimes e_{22}
+\text{\footnotesize$\begin{bmatrix}
\zero & \zero & \zero & \zero & \zero & \zero & \zero & \zero \\
\zero & \zero & \zero & \zero & \zero & \zero & \zero & \zero \\
\zero & \zero & \zero & \zero & \zero & \zero & \zero & \zero \\
\zero & \zero & \zero & \zero & \zero & \zero & \zero & \zero \\
\zero & \zero & \zero & \zero & \zero & \circledast & \circledast & \circledast \\
\zero & \zero & \zero & \zero & \circledast & \zero & \zero & \zero \\
\zero & \zero & \zero & \zero & \circledast & \zero & \zero & \zero \\
\zero & \zero & \zero & \zero & \circledast & \zero & \zero & \zero
\end{bmatrix}$}\otimes (e_{33}+e_{44})
\end{align*}
The circled entries ($\circledast$) are the ones that are not allowed by the standard grading \eqref{eq:Sgamma}.

\subsection{Chamseddine-Connes's Dirac operator}\label{sec:CC}
Let
$$
\hspace*{-3pt}D_0=\!\text{\footnotesize$\left[\begin{array}{cccc|cccc}
\zero & \zero & \bar\Upsilon_\nu\!\! & \zero & \zero & \zero & \zero & \zero \\
\zero & \zero & \zero & \bar\Upsilon_e\!\! & \bar\Omega & \zero & \zero & \zero \\
\Upsilon_\nu\!\! & \zero & \zero & \zero & \zero & \zero & \zero & \zero \\
\zero & \Upsilon_e\!\! & \zero & \zero & \zero & \zero & \zero & \zero \\
\hline
\zero & \Omega & \zero & \zero & \zero & \Delta & \zero & \zero \\
\zero & \zero & \zero & \zero & \Delta & \zero & \zero & \zero \\
\zero & \zero & \zero & \zero & \zero & \zero & \zero & \zero \\
\zero & \zero & \zero & \zero & \zero & \zero & \zero & \zero
\end{array}\right]$}\otimes\;e_{11}
+\text{\footnotesize$\left[\begin{array}{cccc|cccc}
\zero & \zero & \bar\Upsilon_u\!\! & \zero & \zero & \zero & \zero & \zero \\
\zero & \zero & \zero & \bar\Upsilon_d\!\! & \zero & \zero & \zero & \zero \\
\Upsilon_u\!\! & \zero & \zero & \zero & \zero & \zero & \zero & \zero \\
\zero & \Upsilon_d\!\! & \zero & \zero & \zero & \zero & \zero & \zero \\
\hline
\zero & \zero & \zero & \zero & \zero & \Delta & \zero & \zero \\
\zero & \zero & \zero & \zero & \Delta & \zero & \zero & \zero \\
\zero & \zero & \zero & \zero & \zero & \zero & \zero & \zero \\
\zero & \zero & \zero & \zero & \zero & \zero & \zero & \zero
\end{array}\right]$}\otimes\;(1-e_{11}) \,,
$$
where all $\Upsilon$'s and $\Omega$ are complex numbers and $\Delta\in\R$.
The Dirac operator of Chamseddine-Connes \cite{CC97,Con06,CCM06,CM08} is
$$
D_F=D_0+J_FD_0J_F+D_R
$$
with $D_R$ given by \eqref{eq:UpsilonR}, $D_0$ as above, and $\Omega=\Delta=0$.
It is compatible with both gradings of previous sections.

\section{The property (M)}\label{sec:6}

Suppose $H$ is a finite-dimensional complex Hilbert space and $A$, $B$ two (real or complex) unital $C^*$-subalgebras of $\mathrm{End}_{\C}(H)$, that commute one with the other. Let $Z(A)$ be the center
of $A$ and $Z(B)$ be the center of $B$.
Note that
\begin{equation}\label{eq:ZAZB}
A\cap B\subset Z(A)\cap Z(B)\subset A'\cap B'
\end{equation}
and that $Z(A)=Z(A')=A\cap A'$, and similarly for $B$.

Recall that $H$ is a Morita equivalence $A$-$B^\circ$-bimodule if{}f $A=B'$,
which is equivalent to the condition $A'=B$ (by von Neumann Bicommutant Theorem:
$A''=A$ and $B''=B$ in the finite-dimensional case). 
\begin{lemma}\label{lemma:10}
If $H$ is a Morita equivalence $A$-$B^\circ$-bimodule, then the inclusions \eqref{eq:ZAZB} are equalities.
\end{lemma}

\begin{proof}
It follows trivially from $Z(A)=A\cap A'$ and $A'=B$, and similar for $Z(B)$.
\end{proof}

\begin{prop}\label{prop:12}~\vspace{-5pt}
\begin{list}{}{\itemsep=0pt \leftmargin=1em}
\item[i)] If $D_F$ and $\gamma_F$ are as in \S\ref{sec:Sgamma},
the property (M), with or without grading, is \underline{not} satisfied.

\pagebreak

\item[ii)] Let $D_F$ and $\gamma_F$ be as in \S\ref{sec:Ngamma}. If the property (M),
with or without grading, holds then\vspace{-5pt}
\begin{list}{$\bullet$}{\itemsep=0pt \leftmargin=1em}
\item each summand in $D_0$ must have at least one circled coefficient ($\circledast$) different from zero;
\item in each of the first two summands, in both the 1st and 2nd row there must be at least one non-zero element;
\item in the first summand: at least one element in the 5th row must be non-zero and at least
one element in the upper-right block must be non-zero.
\end{list}
\end{list}
\end{prop}

\begin{proof}
It is enough to give the proof for the property (M) with grading (the weaker one).

We apply Lemma \ref{lemma:10} to 
$A=\Cl(A_F)_{\mathrm{e}}$, $B=(A_F^\circ)_{\C}$
and $H=H_F$.
Let $D_0$ and $\gamma_F$ be as in \S\ref{sec:Sgamma} or \S\ref{sec:Ngamma}.
Note that $A$ is generated by $A_F$, $[D_0,A_F]$ 
and $\gamma_F$. Moreover, due to the 1st order condition, $A$ and $B$ are mutually commuting.

Any operator $X\in A_F'\cap (A_F^\circ)'$ commuting with $D_0$ and $\gamma_F$ belongs to $A'\cap B'$
(since it also commute with $[D_0,A_F]$). If we can exhibit such an $X$ and prove that $X\notin Z(B)$,
then the property (M) with grading is not satisfied. Note that $Z(B)$ has elements:
$$
a^\circ=
\begin{bmatrix} 1 & 0 \\ \;0\; & \;0\; \end{bmatrix}\otimes
\left[\!
\begin{array}{c|ccc}
\lambda & \;0\; & \;0\; & \;0\; \\
\hline
\begin{matrix} \;0\; \\ 0 \\ 0 \end{matrix} && \alpha 1_3
\end{array}
\!\right]
+\begin{bmatrix} 0 & \;0\; \\ \;0\; & 1 \end{bmatrix}\otimes\left[\!
\begin{array}{c|c}
\begin{matrix} \;\lambda\; & \;0\;  \\ 0 & \lambda' \end{matrix}  &
\begin{matrix} \;0\; & \;0\;  \\ 0 & 0 \end{matrix} \\
\hline
\begin{matrix} \;0\; & \;0\; \\ 0 & 0 \end{matrix} & \beta 1_2
\end{array}
\!\right] \;,
$$
with $\lambda,\lambda',\alpha,\beta\in\C$.
For $D_0,\gamma_F$ as in \S\ref{sec:Sgamma}, the operator $X=e_{55}\otimes (1-e_{11})$ does the job:\vspace{-2pt}
\begin{itemize}\itemsep=0pt
\item[1)]
it commutes with $D_0$ and $\gamma_F$,
\item[2)]
it belongs to $A_F'=C_F\otimes M_4(\C)$ ($e_{55}\in C_F$: take $q_{22}=1$ and all other coefficients zero in \eqref{eq:CF}),
\item[3)]
it belongs to $(A_F^\circ)'$ (cf.~Lemma \ref{lemma:2}),
\item[4)]
and it does not belong to $Z(B)$.\vspace{-2pt}
\end{itemize}
Let now $D_0$ and $\gamma_F$ be as in \S\ref{sec:Ngamma}.
Concerning the first summand:\vspace{-2pt}
\begin{itemize}\itemsep=0pt
\item if all the circled terms ($\circledast$) are zero, then $X=(e_{66}+e_{77}+e_{88})\otimes e_{11}$ satisfies the conditions (1-4) above;
\item if all the elements in the 1st resp.~2nd row are zero (and then also in 1st resp.~2nd column, by hermiticity), then
$X=e_{11}\otimes e_{11}$ resp.~$X=e_{22}\otimes e_{11}$ satisfies (1-4);
\item if all the elements in the 5th row are zero, similarly by hermiticity $X=e_{55}\otimes e_{11}$ satisfies the conditions (1-4) above;
\item if all the elements in the upper-right blocks are zero, then $X=(e_{11}+e_{22}+e_{33}+e_{44})\otimes e_{11}$ satisfies the conditions (1-4) above.\vspace{-2pt}
\end{itemize}
Concerning the second summand:\vspace{-2pt}
\begin{itemize}
\item if all the elements in the 1st resp.~2nd row are zero (and then also in 1st resp.~2nd column, by hermiticity), then
$X=e_{11}\otimes (1-e_{11})$ resp.~$X=e_{22}\otimes (1-e_{11})$ satisfies the conditions (1-4) above;\vspace{-2pt}
\end{itemize}

\pagebreak

\noindent
Concerning the third resp.~fourth summands:\vspace{-2pt}
\begin{itemize}
\item if the circled terms ($\circledast$) are zero, $X=e_{55}\otimes e_{22}$ resp.~$X=e_{55}\otimes (e_{33}+e_{44})$ satisfies the conditions (1-4) above.
\qedhere
\end{itemize}
\end{proof}

\begin{cor}
Let $D_F$ be as in \S\ref{sec:CC} and $\gamma_F$ one of the two gradings \eqref{eq:Sgamma} or \eqref{eq:Ngamma}. If $\Delta=0$ or $\Omega=0$, the property (M), with or without grading, is not satisfied.
\end{cor}

The operator $D_F$ in \S\ref{sec:CC}, with $\Omega\neq 0$ and $\Delta\neq 0$, represents a minimal
modification of the Dirac operator of \cite{CC97,Con06,CCM06,CM08} which satisfies all the conditions
in Prop.~\ref{prop:12}. We will now show that for such an operator, the Morita condition is satisfied.

\subsection{Morita with a grading}
This section is devoted to prove the following theorem.

\begin{thm}\label{thm:1}
Let $\gamma_F$ be as in \S\ref{sec:Ngamma},
$D_F$ as in \S\ref{sec:CC} with all coefficients different from zero,
and assume that at least one of the following conditions holds:
$$
1.\qquad \Upsilon_\nu\neq \pm\Upsilon_u \;,\hspace{4cm}
2.\qquad \Upsilon_e\neq \pm\Upsilon_d \;.
$$
Then, the spectral triple:\vspace{-3pt}
\begin{itemize}\itemsep=0pt
\item[i)] does not satisfy the property (M);
\item[ii)] it satisfies the property (M) with grading.
\vspace{-3pt}
\end{itemize}
As a corollary, $\Cl(A_F)_{\mathrm{o}}\neq\Cl(A_F)_{\mathrm{e}}$,
so $\gamma_F\notin\Cl(A_F)_{\mathrm{o}}$.
\end{thm}

\noindent
We need a preliminary lemma.
From now on, we assume that the hypothesis of Thm.~\ref{thm:1} are satisfied.

\begin{lemma}\label{lemma:15}
The $A_F$-bimodule $\Omega^1$ is generated by the four elements:
\begin{align*}
\omega_\nu & =e_{31}\otimes\big(\Upsilon_\nu e_{11}+\Upsilon_u (1-e_{11})\big) \;,&
\xi &= e_{52}\otimes e_{11} \;,\\[3pt]
\omega_e & =e_{42}\otimes\big(\Upsilon_e e_{11}+\Upsilon_d (1-e_{11})\big) \;,&
\eta &=e_{56}\otimes 1 \;,
\end{align*}
and their adjoints.
\end{lemma}
\begin{proof}
A linear basis of $(A_F)_{\C}$ is given by the elements:
\begin{align*}
\hspace{2cm}
X_{ij}&:=e_{ij}\otimes 1 \qquad\text{with}\;\;i,j=3,4,
&
Y&:=e_{22}\otimes 1 \;,
\\
Z_{kl} &:=e_{kl}\otimes 1 \qquad\text{with}\;\;k,l=6,7,8,
&
T&:=(e_{11}+e_{55})\otimes 1 \;.
\end{align*}
For any projection $p^2=p=p^*$, the commutator $[D_F,p]=[D_F,p^2]=p[D_F,p]+[D_F,p]p$
is a linear combination of $p[D_F,p]$ and its adjoint $-[D_F,p]p$.
Hence $X_{33}[D_F,X_{33}]$ and $X_{44}[D_F,X_{44}]$ can be taken as generators, instead of $[D_F,X_{33}]$
and $[D_F,X_{44}]$. An explicit computation gives:
\begin{align*}
-X_{33}[D_F,X_{33}] &=\omega_\nu \;, &
-X_{44}[D_F,X_{44}] &=\omega_e \;.
\end{align*}
Note that $[D_F,X_{34}]$ is also the adjoint of $[D_F,X_{43}]$, and
$$
[D_F,X_{43}] =(X_{34}\omega_e)^*-X_{43}\omega_\nu
$$
is still generated by $\omega_\nu$, $\omega_e$ and adjoints. Next
$$
[D_F,Y]Y=\omega_e+\Omega\xi \;,\qquad [D_F,Z_{66}]Z_{66}=\Delta\eta \;.
$$
Since $\Omega,\Delta\neq 0$, this proves that $\xi,\eta\in\Omega^1$.

Furthermore
$
[D_F,Z_{6k}]=\Delta\eta Z_{6k}
$ and
$[D_F,Z_{k6}]=-[D_F,Z_{6k}]^*$ are combinations of $\eta$ and $\eta^*$
for all $k=7,8$, and $[D_F,Z_{jk}]=0$ if $j,k=7,8$. Finally
$$
-T[D_F,T]=\omega_\nu^*+\Omega\xi+\Delta\eta \;,
$$
proving that the elements $\omega_\nu,\omega_e,\xi,\eta$ and their adjoints
are a generating family for $\Omega^1$.
\end{proof}

\begin{proof}[Proof of Theorem \ref{thm:1}]
We now prove that: (i) $\Cl(A_F)_{\mathrm{o}}'\supsetneq (A_F^\circ)_{\C}$ (it is
strictly greater), i.e.~the property (M) is not satisfied.
(ii) $\Cl(A_F)_{\mathrm{e}}'=(A_F^\circ)_{\C}$, i.e.~the property (M) with grading is satisfied.

\medskip

$\Cl(A_F)_{\mathrm{o}}'$ is given by the set of elements in Lemma \ref{lemma:1} that commute
with the generators in Lemma \ref{lemma:15}. A tensor $\sum x_{ij}\otimes e_{ij}$, with each $x_{ij}$ as in
\eqref{eq:CF}, commutes with $\eta$ and $\eta^*$ if{}f $q_{12}=q_{21}=0$ and $q_{22}=\delta$.
Hence, the most general $\phi\in A_F'$ commuting with $\eta,\eta^*$ is:
$$
\phi=
e_{11}\otimes a+
e_{22}\otimes b+
(e_{33}+e_{44})\otimes c+
\left(\textstyle{\sum_{i=5}^8}e_{ii}\right)\otimes d
$$
with $a,b,c,d\in M_4(\C)$ arbitrary matrices. Its commutator with $\xi$ and $\xi^*$
vanishes if{}f
\begin{equation}\label{eq:condA}
de_{11}=e_{11}b \;,\qquad
e_{11}d=be_{11} \;.
\end{equation}
Its commutator with $\omega_\nu,\omega_e$ and their adjoints vanishes if{}f:
\begin{equation}\label{eq:condB}
Ea=cE \;,\qquad
aE=Ec \;,\qquad
Fb=cF \;,\qquad
bF=Fc \;,
\end{equation}
where
$$
E:=\left[\!
\begin{array}{c|ccc}
\Upsilon_\nu & \;0\; & \;0\; & \;0\; \\
\hline
\begin{matrix} \;0\; \\ 0 \\ 0 \end{matrix} && \Upsilon_u1_3
\end{array}
\!\right] \;,\qquad\quad
F:=\left[\!
\begin{array}{c|ccc}
\Upsilon_e & \;0\; & \;0\; & \;0\; \\
\hline
\begin{matrix} \;0\; \\ 0 \\ 0 \end{matrix} && \Upsilon_d1_3
\end{array}
\!\right] \;,
$$
are invertible by hypothesis. It follows from \eqref{eq:condB} that $c$
commutes with both $E^2$ and $F^2$. If the hypothesis of Theorem \ref{thm:1}
are satisfied, at least one of the matrices $E^2,F^2$ is not proportional
to the identity. Its commutation with $c$ implies that
$$
c=\left[\!
\begin{array}{c|ccc}
\lambda & \;0\; & \;0\; & \;0\; \\
\hline
\begin{matrix} \;0\; \\ 0 \\ 0 \end{matrix} && m
\end{array}
\!\right] \;,
$$
for some $\lambda\in\C$ and $m\in M_3(\C)$. But then $c$ commutes with $E$ and $F$ as well,
and it follows from \eqref{eq:condB} that $a=E^{-1}cE=c$ and $b=F^{-1}cF=c$.
Now, $b$ commutes with $e_{11}$ as well, and from \eqref{eq:condA} we get
$$
d=\left[\!
\begin{array}{c|ccc}
\lambda & \;0\; & \;0\; & \;0\; \\
\hline
\begin{matrix} \;0\; \\ 0 \\ 0 \end{matrix} && m'
\end{array}
\!\right] \;,
$$
with the same $\lambda$ as before, and with $m'\in M_3(\C)$.
Thus, $\Cl(A_F)_{\mathrm{o}}'$ has elements
\begin{equation}\label{eq:phi}
\phi=
\begin{bmatrix} 1 & 0 \\ \;0\; & \;0\; \end{bmatrix}\otimes
\left[\!
\begin{array}{c|ccc}
\lambda & \;0\; & \;0\; & \;0\; \\
\hline
\begin{matrix} \;0\; \\ 0 \\ 0 \end{matrix} && m
\end{array}
\!\right]+
\begin{bmatrix} \;0\; & \;0\; \\ 0 & 1 \end{bmatrix}\otimes
\left[\!
\begin{array}{c|ccc}
\lambda & \;0\; & \;0\; & \;0\; \\
\hline
\begin{matrix} \;0\; \\ 0 \\ 0 \end{matrix} && m'
\end{array}
\!\right] \;,
\end{equation}
with $\lambda\in\C$ and $m,m'\in M_3(\C)$, and is strictly greater than $(A_F^\circ)_{\C}$.

\medskip

Imposing the extra condition $[\phi,\gamma_F]=0$, we reduce one $M_3(\C)$ to $\C\oplus M_2(\C)$.
Indeed $[\phi,\gamma_F]=0$ if{}f $d$ commutes with the matrix
$$
\begin{bmatrix} 1_2 \\ & \!\!-1_2 \end{bmatrix} \;,
$$
i.e.~$m'$ belongs to $\C\oplus M_2(\C)\subset M_3(\C)$. This proves that
$\Cl(A_F)_{\mathrm{e}}'=(A_F^\circ)_{\C}$.
\end{proof}

\subsection{Morita without the grading}\label{sec:Mgamma}
Let $D_0$ be as in \S\ref{sec:CC} and $\widetilde{D}_F=\widetilde{D}_0+J_F\widetilde{D}_0J_F+D_R$,
with
$$
\widetilde{D}_0:=D_0+\Gamma(e_{57}+e_{75})\otimes e_{22} \;.
$$
Note that this is still of the type described in \S\ref{sec:Ngamma}. 
Here we have three additional parameters with respect to \cite{CC97,Con06,CCM06,CM08}:
$\Omega\in\C$ and $\Delta\in \R$ as in \S\ref{sec:CC}, and the new one $\Gamma\in\R$.

\begin{thm}\label{thm:2}
Let $\Upsilon_\nu,\Upsilon_e,\Upsilon_u,\Upsilon_d,\Omega,\Delta,\Gamma$ be all different from zero,
and at least one of the following two conditions satisfied:
$$
1.\qquad \Upsilon_\nu\neq \pm\Upsilon_u \;,\hspace{4cm}
2.\qquad \Upsilon_e\neq \pm\Upsilon_d \;.
$$
Then $(A_F,H_F,\widetilde{D}_F,J_F)$ satisfies the property (M).
\end{thm}

\begin{lemma}\label{lemma:16}
If $\Omega,\Delta,\Gamma\neq 0$,
the $A_F$-bimodule $\Omega^1$ is generated by the elements in Lemma~\ref{lemma:15}
plus the element
$$
\zeta=e_{57}\otimes e_{22}
$$
and its adjoint.
\end{lemma}
\begin{proof}
Repeating the proof of Lemma \ref{lemma:15}, the only change is:
$$
[D_F,Z_{77}]Z_{77}=\Gamma \zeta \;,\qquad
[D_F,Z_{78}]=\Gamma \zeta Z_{78} \;,\qquad
-T[D_F,T]=\omega_\nu^*+\Omega\xi+\Delta\eta+\Gamma\zeta \;,
$$
and $[D_F,Z_{87}]=-[D_F,Z_{78}]^*$.
\end{proof}

\begin{proof}[Proof of Theorem \ref{thm:2}]
$\Cl(A_F)_{\mathrm{o}}'$ now is the set of elements
$\phi$ in \eqref{eq:phi} which in addition commute with $\zeta$.
But $[\phi,\zeta]=0$ if{}f $m'\in\C\oplus M_2(\C)\subset M_3(\C)$, so
$\Cl(A_F)_{\mathrm{o}}'=(A_F^\circ)_{\C}$ and the property (M) holds.
\end{proof}

\section{Some remarks on orientability and irreducibility}
\subsection{Orientability}\label{sec:orient}
A classification of finite-dimensional spectral triples satisfying, among other axioms, the orientability condition is in \cite{Kra97}; in fact, in the classification of Dirac operators such assumption plays a crucial role: for example, the uniqueness of the decomposition in \S3.4 follows immediately from the orientability condition. In our case (in the Standard Model with neutrino mixing), the sum $A_F'+(A_F^\circ)'$ is not direct, the term $D_R$ in \eqref{eq:UpsilonR} being an example of non-trivial element in the intersection. This already suggests that the orientability condition in global dimension zero is not satisfied. In fact, we can say something more.

\begin{prop}
Let either \\[2pt]
\hspace*{5pt}(1) $D_F$ and $\gamma_F$ be as in Theorem \ref{thm:1}, or
\\[1pt]
\hspace*{5pt}(2) $\gamma_F$ be the standard grading in \S\ref{sec:Sgamma} and $D_F$ \underline{an\smash{y}} operator
of the type described in \S\ref{sec:Sgamma}.
\\[2pt]
In both cases, there is no chain $c\in A_F^{\otimes n+1}$ such that $\pi_D(c)=\gamma_F$, for any $n\geq 0$.
\end{prop}

\begin{proof}
(1) $\gamma_F\notin\Cl(A_F)_{\mathrm{o}}$, as stated in Theorem \ref{thm:1}.
\\[2pt]
(2) Let $X:=e_{55}\otimes e_{23}$. This operator commutes with $A_F$ and $D_0$, hence
with any element of $\Cl(A_F)_{\mathrm{o}}$.
But it anticommutes with $\gamma_F$, proving that $\gamma_F\notin\Cl(A_F)_{\mathrm{o}}$.
\end{proof}

\noindent
A stronger statement holds for Chamseddine-Connes Dirac operator.

\begin{prop}
Let $\gamma_F$ as in \eqref{eq:Sgamma} or \eqref{eq:Ngamma} and 
$D_F$ as in \S\ref{sec:CC}.
If $\Upsilon_\nu,\Omega,\Delta$ are all zero, then
there is no chain $c\in (A_F\otimes A_F^\circ)\otimes A_F^{\otimes n}$
such that $\pi_D(c)=\gamma_F$, for any $n\geq 0$.
\end{prop}

\begin{proof}
The element $X=e_{15}\otimes e_{11}$ commutes with $A_F$, $A_F^\circ$ and $D_0$, hence with any element in the image of the map $\pi_D$, but it anticommutes with $\gamma_F$, hence $\gamma_F\notin\mathrm{Im}(\pi_D)$.
\end{proof}

\noindent
For $0$-chains, it follows from the argument in \cite{Kra97} that, no matter which Dirac operator one chooses, the orientability conditions cannot be satisfied.

\begin{prop}
For $\gamma_F$ as in \eqref{eq:Sgamma} or \eqref{eq:Ngamma} there is no $c\in A_F\otimes A_F^\circ$ such that $\pi_D(c)=\gamma_F$.
\end{prop}

\begin{proof}
The operator $X:=e_{15}\otimes e_{11}$ belongs to $A_F'\cap (A_F^\circ)'$, but anticommutes with $\gamma_F$. Hence $\gamma_F$ is not in the algebra generated by $A_F$ and $A_F^\circ$.
\end{proof}

For the spectral triple of Theorem \ref{thm:2}, on the other hand, 
since $\gamma_F\in (A_F^\circ)'$ and $(A_F^\circ)'=\Cl(A_F)_{\mathrm{o}}$
due to the property (M), it immediately follows that $\gamma_F\in\Cl(A_F)_{\mathrm{o}}$.

\begin{prop}\label{prop:chain}
Let the spectral triple be as in Theorem \ref{thm:2}. Then there is a $c\in\bigoplus_{n\geq 0}A_F^{\otimes n+1}$
such that $\pi_D(c)=\gamma_F$.
\end{prop}

Of course, this gives no clue on whether $c$ in previous proposition is an simple tensor (so, a chain) or possibly a cycle.


\subsection{Irreducibility}\label{sec:8}

We say that a real spectral triple $(A,H,D,J)$ is \emph{irreducible} if there is no
proper subspace of $H$, other than $\{0\}$, which carries a subrepresentation of $A$
and is stable under $D,J$ and (in the even case) $\gamma$.
Equivalently, it is irreducible if there is no non-trivial projection $p=p^*=p^2\in\mathrm{End}_{\C}(H)$
(so, other than $0$ and $1$), which commute with $A,D,J$ and $\gamma$ in the even case
\cite[Def.~11.2]{GVF01}.

Let $D_F$ be the operator in \S\ref{sec:CC}, and $\gamma_F$ one of the gradings in \eqref{eq:Sgamma}
or \eqref{eq:Ngamma}. If $\Delta=0$ (and possibly $\Omega\neq 0$), then $(A_F,H_F,D_F,\gamma_F,J_F)$
is clearly reducible. Take:
$$
p=\left(\textstyle{\sum_{i=1}^4}e_{ii}\right)\otimes e_{11}+e_{55}\otimes 1
$$
the operator projecting on the subspace of $H_F$ containing only leptons. It clearly commutes with $A_F,D_F,\gamma_F$ and $J_F$.

In order to have irreducibility, we need in $D_F$ a term mixing leptons and quarks.

\begin{prop}
The even spectral triple of Theorem \ref{thm:1} and
the odd spectral triple of Theorem \ref{thm:2} are both irreducible.
\end{prop}

\begin{proof}
For the even spectral triple of Theorem \ref{thm:1}, if $p$ is a projection commuting with $A_F,D_F,\gamma_F,J_F$.
Then it belongs to $\Cl(A_F)_{\mathrm{e}}'=A_F^\circ$ (property (M)). Similarly, for the odd spectral triple of Theorem \ref{thm:2},
one proves that $p$ must belong to $\Cl(A_F)_{\mathrm{o}}'=(A_F^\circ)_{\C}$. But it also commutes with $A_F^\circ=J_FA_FJ_F$,
so it belongs to the center of $(A_F^\circ)_{\C}$. Hence:
$$
p=
\begin{bmatrix} 1 & 0 \\ \;0\; & \;0\; \end{bmatrix}\otimes
\left[\!
\begin{array}{c|ccc}
\lambda & \;0\; & \;0\; & \;0\; \\
\hline
\begin{matrix} \;0\; \\ 0 \\ 0 \end{matrix} && \beta 1_3
\end{array}
\!\right]
+\begin{bmatrix} 0 & \;0\; \\ \;0\; & 1 \end{bmatrix}\otimes\left[\!
\begin{array}{c|c}
\begin{matrix} \;\lambda\; & \;0\;  \\ 0 & \lambda' \end{matrix}  &
\begin{matrix} \;0\; & \;0\;  \\ 0 & 0 \end{matrix} \\
\hline
\begin{matrix} \;0\; & \;0\; \\ 0 & 0 \end{matrix} & \delta 1_2
\end{array}
\!\right] \;,
$$
with $\lambda,\lambda',\beta,\delta\in\C$. Since
$$
J_FpJ_F=\big\{\lambda(e_{11}+e_{55})+\lambda'e_{22}+\delta(e_{33}+e_{44})+\beta (e_{55}+e_{66}+e_{77})\big\}\otimes 1\;,
$$
if $p$ commutes with $J_F$ it is proportional to the identity, hence $p=0$ or $p=1$.
\end{proof}

Let us mention that other inequivalent definitions of irreducibility can be used.
For example, the one adopted in \S18.3 of \cite{CM08} says that a real spectral triple
is irreducible if $H$  carries an irreducible representation of $A$ and $J$. Such a condition
is stronger than the one used by us, and is the condition leading to
the algebra $M_2(\Q)+M_4(\C)$. This is later reduced to $A_F$ (which allows to introduce a grading
and leads to the original Dirac operator of \cite{CC97,Con06,CCM06,CM08}), thus loosing the irreducibility
property. In the next section, we discuss the intermediate algebra $A^{\mathrm{ev}}$ of the Pati-Salam model.

\subsection{On the Pati-Salam model}\label{sec:9}
The Pati-Salam model is a grand unified theory with gauge group $\mathrm{Spin}(4)\times\mathrm{Spin}(6)\simeq SU(2)\times SU(2)\times SU(4)$.
The relevant algebra is now $A^{\mathrm{ev}}=\Q\oplus\Q\oplus M_4(\C)$, which we identify with 
the subalgebra of elements $a\otimes 1\in\mathrm{End}_{\C}(H_F)$, with $a$ of the form:
$$
a=\left[\!
\begin{array}{c|c}
\!\left[\begin{array}{c|c} x \\ \hline & y \end{array}\right]\!
&
\phantom{\!\left[\begin{array}{c|c} x \\ \hline & y \end{array}\right]\!}
\\[14pt]
\hline
\begin{matrix} ~ \\[7pt] ~ \end{matrix} & m
\end{array}\!\right] \;,
$$
with $x,y\in\Q$ (and we think of them as $2\times 2$ complex matrices) and $m\in M_4(\C)$.
All off-diagonal blocks are zero.

The data $(A^{\mathrm{ev}},H_F,J_F,\gamma_F,D_F)$, with $D_F$ as in \S\ref{sec:CC}
and $\gamma_F$ as in \eqref{eq:Sgamma}, satisfies all the conditions of a real spectral triple
except for the 1st order condition \cite{CCvS13a,CCvS13b} (and then the property (M) cannot be satisfied).
On the other hand, it is a simple check to verify that irreducibility, in the stronger sense of \S18.3 of \cite{CM08}
(so, without $\gamma_F$ and $D_F$) is satisfied.

\begin{lemma}
The commutant $(A^{\mathrm{ev}})'$ has elements $\sum a\otimes b$ with $b\in M_4(\C)$ arbitrary and
$a\in M_8(\C)$ of the form
\addtolength{\arraycolsep}{3pt}
$$
a=\left[\!
\begin{array}{cccc}
\multicolumn{1}{c|}{\rule{0pt}{18pt}\alpha 1_2\!} \\[5pt]
\cline{1-2}
& \multicolumn{1}{|c|}{\rule{0pt}{18pt}\beta 1_2\!} && \\[5pt]
\cline{2-4}
\begin{matrix} ~ \\[5pt] ~ \end{matrix} && \multicolumn{2}{|c}{\hspace{8pt}\delta 1_4\hspace{5pt}}
\end{array}\!\right] \;,
$$
where the $\alpha$ and $\beta$-blocks are $2\times 2$, the $\delta$-block is $4\times 4$, and
$\alpha,\beta,\delta\in\C$.
\end{lemma}

\begin{proof}
By direct computation.
\end{proof}

\begin{prop}
There is no non-trivial projections on $H_F$ commuting with $A^{\mathrm{ev}}$ and $J_F$.
\end{prop}

\begin{proof}
It follows from previous lemma that any $p$ commuting with $A^{\mathrm{ev}}$ has the form
$$
p=(e_{11}+e_{22})\otimes\alpha+(e_{33}+e_{44})\otimes\beta+\left(\textstyle{\sum_{i=5}^8}e_{ii}\right)\otimes\delta \;,
$$
where now $\alpha,\beta,\delta\in M_4(\C)$ are three projections. Since (in $4\times 4$ blocks):
$$
J_FpJ_F=
\begin{bmatrix} \;0\; & \;0\; \\ 0 & \alpha \end{bmatrix}\otimes (e_{11}+e_{22})+
\begin{bmatrix} \;0\; & \;0\; \\ 0 & \beta \end{bmatrix}\otimes (e_{33}+e_{44})+
\begin{bmatrix} \delta & 0 \\ \;0\; & \;0\; \end{bmatrix}\otimes 1 \;,
$$
we deduce that $p$ commutes with $J_F$ if and only if $\alpha=\beta=\delta$ are proportional to the identity,
and then $p=0$ or $p=1$ is a trivial projection.
\end{proof}

Orientability (in the weak sense) is also easy to check, since in \eqref{eq:Sgamma} the first summand belongs to $A^{\mathrm{ev}}$ and the second is minus the conjugated by $J_F$. So,
$\gamma_F\in A^{\mathrm{ev}}+J_FA^{\mathrm{ev}}J_F$ (which implies weak orientability, since every $0$-chain is a cycle).


\section{Conclusions}\label{sec:10}
In this paper we studied the property (M), cf.~Def.~\ref{df:propM}, which is a possible natural mathematical generalization of the notion of spin-manifold and of Dirac spinors to noncommutative geometry.
Although the original Chamseddine-Connes' spectral triple is shown not to satisfy this property, we find that it is enough to add two terms to the Dirac operator $D_F$ and slightly change the grading in order to satisfy it.
The new terms in the Dirac operator will generate of course new fields (they are introduced with the purpose of enlarging the module of $1$-forms, and then the Clifford algebra; 
generators are given in Lemma \ref{lemma:15}). 
Then study of whether (and how) they contribute to the action functional of the model is however beyond the scope of this paper. Obviously the non zero (for the property (M)) constants $\Omega, \Delta$ in front of the new terms can be arbitrarily small and so below the current experimental observation threshold.
Their fine-tuning in order to get the correct value of the Higgs mass could be studied in future work.

Of the two terms, the $\Omega$-term is compatible with both the original and the modified grading of \S\ref{sec:Sgamma} and \S\ref{sec:Ngamma}. The $\Delta$-term on the other hand is compatible (anticommutes) only with the modified grading. Such a term may potentially (see \cite{PSS97}) mix quarks and leptons, and although it may seem exotic, it is also necessary for the  irreducibility of the spectral triple (cf.~\S\ref{sec:8}): without this term, the leptonic and quark sectors of $H_F$ carry each one a sub-spectral triple.
A third additional term (with coefficient $\Gamma$) in $D_F$ is instead necessary (though non sufficient) if one wants the spectral triple to be also orientable,  cf.~Prop.~\ref{prop:chain} for the precise statement.

Concerning the grading, the one in \eqref{eq:Ngamma} is minimal a modification the one in \eqref{eq:Sgamma}: they agree on leptons and have opposite sign on baryons (quarks).
A study of the physical consequences of this modification are beyond the scope of this paper. The internal grading contributes to the grading of the full spectral triple, product of the finite-dimensional one with the canonical spectral triple of a 4-dimensional spin manifold. The full grading is used 
to project out from the Hilbert space unphysical degrees of freedom (and partially solve the quadrupling of degrees of freedom, cf.~\cite{CM08} for the details). Thus, changing the grading in principle could affect this part of the theory, which should be studied in the future.

\smallskip

We close by stressing again that the aim of this paper was a mathematical study of the property (M), and few related issues, but a detailed analysis of the physical aspects of our model (for example, understanding what happens to the Higgs mass) goes beyond the scope of the paper and is postponed to future works.

\subsection*{Acknowledgments}\vspace{-3pt}
F.D.~was partially supported by UniNA and Compagnia di San Paolo in the framework of the program STAR 2013.
L.D.~was partially supported by the Italian grant PRIN 2010 ``Operator Algebras, Noncommutative Geometry and Applications'' and by the Polish grant ``Harmonia''.


\end{document}